\documentclass[11pt,fleqn]{article}

\usepackage{makeidx, natbib, vmargin, setspace}
\usepackage{amssymb, amsthm, amsmath}
\usepackage[colorlinks=true, pdfstartview=FitV, citecolor=blue!50!black, linkcolor=red!30!black]{hyperref}
\usepackage{graphicx, graphics, tikz,booktabs}

\newcounter{llst}
\newenvironment{abet}{\begin{list}{\rm (\alph{llst})}{\usecounter{llst}
\setlength{\itemindent}{0em} \setlength{\leftmargin}{3em}
\setlength{\labelwidth}{2em} \setlength{\labelsep}{1em}}}{\end{list}}
\newenvironment{numm}{\begin{list}{\rm (\roman{llst})}{\usecounter{llst}
\setlength{\itemindent}{0em} \setlength{\leftmargin}{3.5em}
\setlength{\labelwidth}{2.5em} \setlength{\labelsep}{1em}}}{\end{list}}

\newtheorem{theorem}{Theorem}[section]

\newtheorem{corollary}[theorem]{Corollary}

\newtheorem{definition}[theorem]{Definition}
\newtheorem{expl}[theorem]{Example}

\newtheorem{lemma}[theorem]{Lemma}

\newtheorem{proposition}[theorem]{Proposition}

\newtheorem{rmrk}[theorem]{Remark}

\newenvironment{example}{\begin{expl} \rm}{\hfill $\blacklozenge$ \end{expl}}{}
{}

\begin{document}

\title{\textbf{Middlemen and Contestation \\ in Directed Networks}}
\author{Owen Sims\footnote{Centre for Data Science and Scalable Computing, Institute of Electronics, Communications and Information Technology (ECIT), Queen's University Belfast, Northern Ireland Science Park, Queens Road, Belfast BT3 9DT, UK. Email: \href{osims01@qub.ac.uk}{osims01@qub.ac.uk}} 
\and Robert P. Gilles\footnote{Queen's University Management School, Riddel Hall, 185 Stranmillis Road, Belfast BT9 5EE, UK. Email: \href{r.gilles@qub.ac.uk}{r.gilles@qub.ac.uk}}
}
\date{December 2016}
\maketitle

\begin{abstract}
\singlespace\noindent
This paper studies middlemen---or ``critical nodes''---that intermediate flows in a directed network. The contestability of a node is introduced as a network topological concept of competitiveness meaning that an intermediary's role in the brokering of flows in the network can be substituted by a group of other nodes. We establish the equivalence of uncontested intermediaries and middlemen.

The notion of node contestability gives rise to a measure that quantifies the control exercised by a middleman in a network. We present a comparison of this middleman centrality measures with relevant, established network centrality measures. Furthermore, we provide concepts and measures expressing the \emph{robustness} of a middleman as the number of links or nodes that have to be added to or removed from the network to nullify the middleman's power.

We use these concepts to identify and measure middleman power and robustness in two empirical networks: \emph{Krackhardt's advice network} of managers in a large corporation and the well-known \emph{Florentine marriage network} as a proxy of power brokerage between houses in Renaissance Florence.
\end{abstract}

\bigskip
\begin{singlespace}
\end{singlespace}

\thispagestyle{empty}

\newpage

\setcounter{page}{1} \pagenumbering{arabic}

\section{The critical role of middlemen in networks}

Research has recognised networks as important descriptors of social and economic processes \citep{Watts2004, Jackson2008, Newman2010, Barabasi2016}. Node centrality aims to identify the most influential nodes in the network. This paper investigates those nodes that are critical for the flow of information and trade in a network. Such critical nodes---or ``middlemen''---have traditionally been considered in undirected networks only.\footnote{Middlemen are defined as so-called \emph{cut nodes} in mathematical graph theory \citep{Wilson2010}. The importance of these critical nodes as conduits of information flows in social networks has been recognised by \citet{Burt1992,Burt2004,Burt2005} and \citet{Burt2010} and in economic networks by \citet{KalaiMiddlemen1978,RubinsteinWolinsky1987,JacksonWolinsky1996,GillesChakrabarti2006} and \citet{Siedlarek2015}. }

Here, we consider these middlemen in the more general context of directed networks. We introduce a middleman as a node that can block the information flow from at least one node to another. If one applies this definition to undirected networks, one arrives at the standard notion of a middleman as a singleton node cut set. In the context of directed networks this is no longer the case: Its removal does not necessarily break up the whole network; it just compromises the information flow for at least one pair of nodes in the network. The brokerage function of middlemen allows them to be highly extractive to both directly and indirectly connected nodes \citep{KalaiMiddlemen1978}.

Naturally, the existence of middlemen is closely related to the interactive environment represented by that network. \citet{GillesDiamantaris2013} show that middlemen have the potential to be highly exploitive given a lack of alternative pathways along which to conduct the required socio-economic interactions. The main conclusion from this research is a non-trivial extension of how economists and social scientists perceive the architecture and dynamics of exchange systems, how the presence of a middleman can hold a system together, and, as a consequence of their position, can act as rent-extracting monopolists with excessive bargaining power \citep[Chapter~11]{EasleyKleinberg2010}.

We reduce these notions to a general definition of ``contestation'' in directed networks. Our notion of contestation refers to a network topological property that indicates whether interaction can be conducted without the involvement of a certain node. Hence, a node is contested if alternative pathways are available to establish interaction between pairs of nodes in the network. Our main result shows that there is a formal duality between the existence of middlemen and network contestability. In particular, an intermediary node is a middleman if and only if it is uncontested.

Despite the wide acknowledgement within social network analysis of the significance of middlemen, centrality measures do not necessarily identify these critical nodes as being important even though their removal may deteriorate the functionality of the network as a whole. As the notion of centrality came to the fore, \citet[p.~219]{Freeman1979} argued that central nodes were those ``in the thick of things''. To exemplify this, he used an undirected star network consisting of five nodes. The middle node, at the centre of the star, has three advantages over the other nodes: It has more ties; it can reach all the others more quickly; and it controls the flow between the others.

Given this sentiment, one would expect that betweenness centrality measures \citep{Freeman1977} would capture the influence that such middlemen exert. However, we show that this is not the case. Instead, we propose a new \emph{middleman power measure} that exhibits the desired properties. We apply this middleman power measure to study two very well known (historical) directed networks from the literature. These applications allow us to make an in-depth comparison with betweenness centrality and eigenvector centrality \citep{Bonacich1987}.

In \citet{Krackhardt1987}'s advice network, our middleman power measure confirms Krackhardt's original assessment of the most influential node. Also, for the Florentine marriage network of the early Renaissance \citep{Padgett1993, Padgett1994}, we conclude that our middleman power measure clearly ranks the more powerful middlemen higher than the less powerful, confirming with the reported historical analysis.  Furthermore, we show that, despite the importance of middlemen in these networks, this positional feature is not properly and fully identified by conventional centrality measures.

Finally, we consider the \emph{robustness} of middleman positions in the network in light of potential changes to the network topology. The more robust a middleman's position, the more changes in the network topology are required to reduce the position of that middleman. Robustness, therefore, refers to a dynamic property of the network topology.

In particular, we look at two fundamentally different measurements of middleman robustness. First, we consider changes to the arcs and links that are present in the network. Thus, the robustness of a middleman is the number of arcs that need to be deleted from or added to the network in order to contest the middleman. Second, we define a node-based robustness concept that counts the number of nodes that need to be removed from the network to contest the middleman. We show that these robustness measures are closely related and that the arc-based measure is sufficient to indicate the robustness of a middleman.

We apply these robustness measures to the Florentine marriage network and identify that the Medici house had a less robust position than considered in the literature. In particular, the Pazzi is shown to have a more robust position, supporting an alternative explanation of the rivalry between these two houses in renaissance Florence.

\paragraph{Outline.}

We follow this introduction with Section~\ref{sec:NetworkPreliminaries}, which discusses the required notions of network science. Section~\ref{sec:MiddlemanContestability} introduces the notions of strong and weak middlemen and considers the dual notion of network contestability. Section~\ref{sec:middlemanPower} discusses our measure of middleman power, which assigns a quantitative expression to a node's brokerage power. Additionally, the section provides three measures regarding the robustness of middleman positions in the network. Section~\ref{sec:empiricalNetwork} investigates two empirical case studies of social networks where middlemen have been identified as critical: Krackhardt's advice network and the Florentine marriage network. Section~\ref{sec:Conclusion} concludes.

\section{Network preliminaries}
\label{sec:NetworkPreliminaries}

In this paper we focus on networks as representations of relational infrastructures that map how entities such as individuals and firms communicate and exchange information, money or goods. These entities are represented by nodes, while relationships are represented by directed arcs between nodes. Throughout we consider these infrastructures as carriers of flows of physical commodities in a logistics framework or of information. Middlemen are identified as those entities that exercise control over the flows in the network.
\begin{definition}
	A (directed) \textbf{network} is a pair $(N,D)$ where $N = \{1,2, \ldots ,n\}$ is a finite set of \textbf{nodes} and $D \subset \{ (i,j) \mid i,j \in N$ and $i \neq j \}$ is a set of \textbf{arcs}, being directed relationships from one node to another.
\end{definition}

\noindent
Throughout, we denote an arc $(i,j) \in D$ by $ij \in D$ and a directed network $(N,D)$ by $D$ unless $N$ is ambiguous. We interpret the nodes $N$ in the network as a set of entities that are connected through the relational infrastructure represented by the arcs in $D$. Thus, $D$ describes the infrastructure within which the nodes are embedded. We point out that $D$ is \emph{irreflexive} in the sense that $(i,i) \notin D$ for any $i \in N$.

We assume throughout that the infrastructure represented by $D$ supports flows of information, money and/or goods between the nodes in $N$ and that the control of these flows is of critical importance to the functionality of this infrastructure.\footnote{We emphasise that we can generalise this setting by assuming that each arc $ij \in D$ has a certain capacity that limits the flow on that arc. In this paper we limit ourselves to the case in which these arcs have unlimited capacity and there are no bounds on the flows in the network. Hence, either a connection between two nodes $i$ and $j$ exists or not.}

An \emph{undirected} network is defined as a network $(N,D)$ such that all arcs are reciprocated: $ij \in D$ if and only if $ji \in D$. An undirected network represents a special case of a network in that there exist arbitrary flows between any two linked nodes.

\paragraph{Walks and paths.}

A \emph{walk} from node $i$ to node $j$ in a network $D$---or simply an $ij$-walk in $D$---is an ordered list of nodes $\Omega_{ij} (D) = \left( i_1, i_2 , \ldots ,i_m \right)$ such that $m \geqslant 3$, $i_1 =i$, $i_m =j$ and $i_k i_{k+1} \in D$ for all $k \in \{ 1, \ldots ,m-1 \}$. Clearly, a walk is a sequence of adjacent nodes in the network. Walks might revisit nodes and, therefore, might contain loops.

The notion of a path excludes loops and can be formulated using set-theoretic notions. Formally, a \textit{path} from $i$ to $j$ in a network $D$---or, simply, an $ij$-\emph{path}---is a finite subset of nodes $W_{ij} (D) = \left\{ i_{1}, \ldots ,i_{m} \right\} \subset N$ with $m \geqslant 3$, $i_1 =i$, $i_m =j$ and $i_{k}i_{k+1} \in D$ for every $k=1, \ldots ,m-1$. Therefore, to every $ij$-path $W_{ij} (D) = \{ i_1, \ldots ,i_m \}$ there corresponds a unique $ij$-walk given by $\Omega_{ij} (D) = ( i_1, \ldots ,i_m )$. We remark that since each element in a set is listed only once, every path is a walk without any loops.

In many cases there are multiple paths from $i$ to $j$ in a network $D$. If this is required, we denote $W_{ij}^{v}(D)$ as the $v$-th distinct path from $i$ to $j$ in $D$. This gives rise to the collection $\mathcal{W}_{ij}(D)= \left\{ W_{ij}^{1}(D), \ldots ,W_{ij}^{V}(D) \right\}$, consisting of all distinct paths $W_{ij} (D)$ from node $i$ to node $j$ in $D$, where $V$ is the number of distinct $ij$-paths in $D$. Here, $\mathcal{W}_{ij}(D)= \varnothing$ denotes that there is no path from node $i$ to node $j$ in the network $D$, while $\mathcal{W}_{ij}(D) \neq \varnothing$ indicates that there exists at least one path from $i$ to $j$ in $D$.

\paragraph{Connectedness.}

We now say that node $i$ is \textit{connected to} node $j$ if $\mathcal{W}_{ij}(D) \neq \varnothing$, i.e., there is at least one path from node $i$ to node $j$ in the network $D$. A network $D$ is \emph{weakly connected} if for all $i,j \in N$ either $i$ is connected to $j$, or $j$ is connected to $i$, or both.

Two distinct nodes $i,j \in N$ are \textit{strongly connected} in the network $D$ if $\mathcal{W}_{ij}(D) \neq \varnothing$ as well as $\mathcal{W}_{ji}(D) \neq \varnothing$, indicating that there is always a path from node $i$ to node $j$ and back. A network $D$ is \emph{strongly connected} if $\mathcal{W}_{ij} \neq \varnothing$ and $\mathcal{W}_{ji} \neq \varnothing$ for all nodes $i,j \in N$. Hence, any strongly connected network is always weakly connected. This implies all nodes are bi-connected in the sense that bilateral flows exists between any two nodes.

A subset of nodes $M \subset N$ is a \emph{component} of $D$ if $(M, D_M)$ with $D_M = D \cap (M \times M)$ is weakly connected and and there is no node $i \in N \setminus M$ outside node set $M$ such that $(M+i, D_{M+i})$ is weakly connected.\footnote{Here we employ that notation that $S+t = S \cup \{ t \}$ for every set $S$ and $t \notin S$.} Clearly, every network contains components. In fact, a weakly connected network consists of exactly one component, namely $D$ itself. The next insight states the converse of this property.
\begin{lemma}
	A network $D$ on node set $N$ is \emph{not} weakly connected if and only if it consists of two or more components.
\end{lemma}

\noindent
Similarly, a subset of nodes $M \subset N$ is a \emph{strong component} of $D$ if $(M, D_M)$ with $D_M = D \cap (M \times M)$ is strongly connected and and there is no node $i \in N \setminus M$ outside node set $M$ such that $(M+i, D_{M+i})$ is strongly connected. Clearly, every strong component of a network $D$ is indeed a component of $D$, but it is not necessarily the case that a component is a strong component.

\paragraph{Successors and predecessors.}

Connectedness in a network leads to several auxiliary concepts. We denote by
\begin{equation}
	s_i (D) = \{ j \in N \mid ij \in D \}
\end{equation}
the set of \emph{direct successors} of node $i$ in network $D$. Similarly, we denote by
\begin{equation}
	p_i (D) = \{ j \in N \mid ji \in D \} \equiv \{ j \in N \mid i \in s_j (D) \}
\end{equation}
the set of \emph{direct predecessors} of node $i$ in network $D$. We note that nodes $h \in p_i (D) \cap s_i (D)$ are the ones that are connected to node $i \in N$ by an undirected link. Hence, $D$ is an undirected network if and only if $s_i (D) = p_i (D)$ for all nodes $i \in N$.

Furthermore, a node $j$ is called a \emph{successor} of node $i$ in network $D$ if $\mathcal{W}_{ij} (D) \neq \varnothing$. Similarly, node $i$ is called a \textit{predecessor} of node $j$ in network $D$ if $j$ is a successor of $j$ in $D$, i.e., $\mathcal{W}_{ij} (D) \neq \varnothing$. This gives rise to the following two node sets in relation to some node $i \in N \colon$
\begin{gather}
	S_i (D)= \left\{ j \in N \, \left| \, \mathcal{W}_{ij}(D) \neq \varnothing \right. \right\} \\[1.5ex]
	P_i (D)= \left\{ j \in N \, \left| \, \mathcal{W}_{ji}(D) \neq \varnothing \right. \right\}
\end{gather}
$S_i$ is the set of successors of node $i$ in network $D$, while $P_i$ is the set of predecessors of node $i$ in network $D$. We note that $i \notin S_i(D)$ as well as $i \notin P_i (D)$. Obviously, $s_i (D) \subset S_i (D)$ and $p_i (D) \subset P_i (D)$. Node $k$ is an \emph{indirect successor} of $i$ in $D$ if $k \in S_{i}(D) \setminus s_i (D)$ meaning that $\mathcal{W}_{ik}(D) \neq \varnothing$. Thus, $i$'s successor set $S_i$ is composed of all of $i$'s direct and indirect successors.

The fact that a node is not a member of its successor and predecessor sets gives rise to the following two definitions.
\begin{gather}
	\overline{S}_{i}(D) = S_{i}(D) \cup \{i\} \\[1.5ex]
	\overline{P}_{i}(D) = P_{i}(D) \cup \{i\}
\end{gather}
We denote by $\overline{S}_{i}(D)$ the \emph{reach} of node $i$ in network $D$, while $\overline{P}_{i}(D)$ denotes node $i$'s \emph{origin} in network $D$. These auxiliary concepts are required in later analysis in this paper.

For node $i \in N$ its \emph{out-degree} in $D$ is defined by $d_{i}^{+} = \# s_{i}(D)$ and its \emph{in-degree} in $D$ by $d_{i}^{-}= \# p_{i}(D)$. Now, $d_{i} = \# \{ s_{i} (D) \cup p_{i} (D) \}$ denotes node $i$'s (overall) \emph{degree} in network $D$.\footnote{Here, $\# S$ denotes the cardinality of the finite set $S$.} Clearly, $d_i \leqslant \min \{ d_i^+ + d_i^- , n-1 \}$.

The node set $N$ can now be partitioned into four distinct subsets of nodes. The classification of nodes is based on their connectivity properties, which gives rise to dividing the nodes into sources, sinks, leaves and intermediaries. This is formalised as follows:
\begin{definition}
	Let $D$ be a directed network on node set $N$ and let $i \in N$ be some node.
	\begin{numm}
		\item Node $i$ is a \textbf{source} in network $D$ if $d_{i}^{-} = 0$ and $d_{i}^{+} \geqslant 1$. We denote
		\[
		N^+_D = \{ i \in N \mid i \mbox{ is a source in } D \, \} .
		\]
	\item Node $i \in N$ is a \textbf{sink} in network $D$ if $d_{i}^{-} \geqslant 1$ and $d_{i}^{+} = 0$. We denote
	\[
		N^-_D = \{ i \in N \mid i \mbox{ is a sink in } D \, \} .
	\]
	\item Node $i \in N$ is a \textbf{leaf} in network $D$ if $d_i^+ = d_i^- = d_i =1$, i.e., $i$ is connected to exactly with one other node through an undirected link. We denote
	\[
	L_D =  \{ i \in N \mid i \mbox{ is a leaf in } D \, \} .
	\]
	\item Node $i \in N$ in an \textbf{intermediary} in network $D$ if $d_i \geqslant 2$ and $d_{i}^{-} \geqslant 1$ as well as $d_{i}^{+} \geqslant 1$. We denote
	\begin{align*}
		M_D & = \{ i \in N \mid i \mbox{ is an intermediary in } D \, \} = \\
		& = \{ i \in N \mid \mbox{There exist } j \in p_i (D) \mbox{ and } h \in s_i (D) \mbox{ with } j \neq h \, \} .
	\end{align*}
	\end{numm}
\end{definition}

\noindent
We remark that the node set is now partitioned into these constructed subsets:
\begin{equation}
	N = N^+_D \cup N^-_D \cup L_D \cup M_D
\end{equation}
where these constituting sets are pairwise disjoint.

\begin{definition}
	Let $D$ be a network on node set $N$. For every intermediary $i \in M_D$ we denote by
	\begin{equation}
		D - i = D_{N \setminus \{ i \} } = \left\{ \, jh \in D \mid j,h \in N \setminus \{ i \} \, \right\} .
	\end{equation}
	the \textbf{reduced network} after removing intermediary $i$ from the constituting node set.
\end{definition}

\noindent
Therefore, $D-i$ is the restricted network that results after the removal of the intermediary $i$ and all arcs to and from node $i$ from the network $D$.

\section{Middlemen and contestability}
\label{sec:MiddlemanContestability}

In this section we introduce the notion of a middleman in a (directed) network and discuss the notion of contestation as an equivalent, dual conception of these middlemen.

\subsection{Defining middlemen}

We identify \textit{middlemen}---or ``critical nodes''---as intermediary nodes $i \in M_D$ that have the ability to broker certain flows in network $D$. Following the established literature in graph theory and network analysis, a critical node in an undirected network is equivalent to the graph theoretical notion of a \emph{cut node}. Such nodes can severely disrupt and manipulate the typical operations on a network by disconnecting the network into two or more components \citep{KalaiMiddlemen1978,Burt1992,JacksonWolinsky1996,GillesChakrabarti2006}. Here, we extend this concept to arbitrary (directed) networks. In our general context, our conception focusses on the disruption of the connectivity of two or more nodes.

\begin{definition} \label{middleman}
Let $D$ be a network on node set $N$ and let $h \in M_D$ be an intermediary in $D$.
\begin{abet}
\item Node $h$ is an \textbf{$ij$-middleman} in the network $D$, where $i,j \in N$ and $i \neq j$, if
\begin{equation} \label{ijmiddleman}
h \in \cap \mathcal{W}_{ij}(D) \setminus \{ i,j \} = \left( \, W_{ij}^{1}(D) \cap  \ldots  \cap W_{ij}^{V}(D) \, \right) \setminus \{ i,j \}
\end{equation}
where $V \geqslant 1$ is the number of distinct paths from $i$ to $j$. Now, $M_{ij}(D) \equiv \cap \mathcal{W}_{ij}(D) \setminus \{ i,j \}$ denotes the set of all $ij$-middlemen in $D$.

\item An intermediary $h \in M_D$ is a \textbf{middleman} in network $D$ if there exist two distinct nodes $i,j \in N$ with $i \neq j$ such that $h$ is an $ij$-middleman in $D$. \\ The set of all middlemen in $D$ is therefore given by
\begin{equation} \label{middlemanseteq}
\mathbf{M} (D) = \bigcup_{i,j \in N \colon i \neq j} M_{ij}(D)
\end{equation}
\end{abet}
\end{definition}

\noindent
A middleman in a network is an intermediary node that is a member of \emph{all} paths between at least two other nodes within the given network. Therefore, a middleman controls the flow between at least two other nodes. Conversely, a non-middleman is an intermediary that if removed from the network does not affect the connectivity of any two or more other nodes.

The following properties can be deduced directly from the above definition. We state most of these assertions without proof.

\begin{proposition}\label{TheoremIntermediary}
Let $D$ be a directed network on the node set $N = \{ 1, \ldots ,n \}$.
\begin{numm}
	\item If $n \leqslant 2$, there do not exist any middlemen in $D$.
	\item For all $i,j \in N$ with $j \in s_{i}(D)$ it holds that $M_{ij}(D) = \varnothing$. This implies that the complete directed network has no middlemen.
	\item Every middleman $i \in \mathbf M(D)$ has a local clustering co-efficient of less than $1$ in the sense of \citet[Sections 2.10 and 3.9]{Barabasi2016}.
	\item If $D$ is undirected in the sense that $ij \in D$ if and only if $ji \in D$, then $M_{ij} (D) = M_{ji} (D)$ for all $i,j \in N$.
\end{numm}
\end{proposition}

\begin{proof}
	For (iii), note that, if $i \in N$ is a middleman in $D$, then there exist at least one $j \in p_i (D)$ and $h \in s_i (D)$ such that $jh \notin D$ as well as $hj \notin D$. Hence, $i$, $j$ and $h$ do not form a triad in the underlying undirected network and, therefore, the local clustering coefficient of $i$ in $D$ is strictly less than 1.
\end{proof}

\noindent
From (iv) it follows that in an undirected network a node is a middleman if it rests on all paths from node $i$ to node $j$. This means that the removal of a middleman disrupts the communication between nodes $j$ and $i$. Hence, in an undirected network, a middleman is indeed a ``cut node'' or a singleton ``cut set'' as traditionally understood \citep[\S 5]{Wilson2010}.

By definition a middleman rests on all paths from node $i$ to node $j$, but does not have to rest on all paths from $j$ to $i$. This implies that the removal of a middleman in a network disrupts the interaction from $i$ to $j$, but the interaction from node $j$ to node $i$ may be unaffected. Hence, even with the removal of a middleman from a network, nodes $i$ and $j$ can still remain weakly connected. This implies in particular that the removal of a middleman might leave the network weakly connected, although not strongly connected. This insight motivates a further refinement of the notion of a middleman in a directed network.

\begin{definition} \label{strongweakmiddlemen} 
Let $D$ be a weakly connected network on node set $N = \{ 1, \ldots ,n \}$ with $n \geqslant 3$.
\begin{numm}
	\item A middleman $h \in \mathbf M (D)$ is a \textbf{strong middleman} in $D$ if the reduced network $D - h$ is not weakly connected and consists of two or more components.

\item A middleman $h \in \mathbf M (D)$ is a \textbf{(regular) middleman} in $D$ if the reduced network $D - h$ is weakly connected.
\end{numm}
\end{definition}

\noindent
Regular middlemen exist in both cyclic and acyclic directed networks. Example~\ref{identifyingmiddlemen} highlights the existence of both weak and strong middlemen in an acyclic network.

\begin{figure}[h]
\begin{center}
\begin{tikzpicture}[scale=0.5]
\draw[thick, ->] (0,2.5) -- (4.4,4.5);
\draw[thick, ->] (0,2.5) -- (4.4,0.4);
\draw[thick, ->] (5,5) -- (9.3,5);
\draw[thick, ->] (5,0) -- (9.3,0);
\draw[thick, ->] (5,5) -- (9.3,0.5);
\draw[thick, ->] (10,5) -- (14.3,2.8);
\draw[thick, ->] (10,0) -- (14.3,2.2);
\draw[thick, ->] (15,2.5) -- (19.3,2.5);

\draw (0,2.5) node[circle,fill=black!10] {$1$};
\draw (5,5) node[circle,fill=black!25] {$2$};
\draw (5,0) node[circle,fill=black!10] {$3$};
\draw (10,5) node[circle,fill=black!10] {$4$};
\draw (10,0) node[circle,fill=black!25] {$5$};
\draw (15,2.5) node[circle,fill=black!25,draw,very thick] {$6$};
\draw (20,2.5) node[circle,fill=black!10] {$7$};

\end{tikzpicture}
\end{center}
\caption{Acyclic network highlighting weak and strong middlemen}
\label{weakmm}
\end{figure}
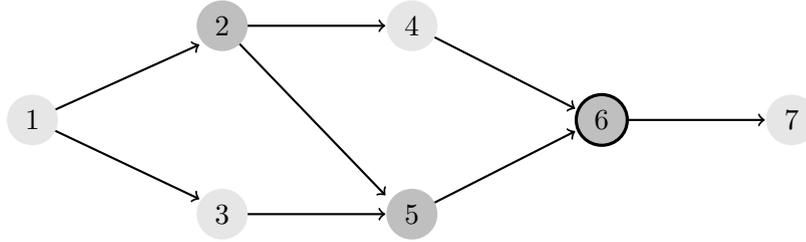

\begin{example} \label{identifyingmiddlemen}
Consider a weakly connected network $D$ on the node set $N=\{1,2,3,4,5,6,7\}$, which is depicted in Figure~\ref{weakmm}. This network represents a flow system from the left to the right. We easily determine that $\mathbf M (D)=\{2,5,6\}$, where nodes 2 and 5 are regular middlemen and node 6 is a strong middleman. In the graphical representation this is indicated by the shading of these nodes and the circling of node $6$.
\\
Consider node 2. Node 2 lies on all paths from node 1 to node 4, therefore $2 \in M_{1,4}(D)$ and $\mathcal{W}_{1,4}(D - 2) = \varnothing$. However, the reduced network $D - 2$ remains weakly connected, meaning that node 2 must be a regular middleman.
\\
An analogous argument could be made for node 5 because $5 \in M_{3,6} (D) \cap M_{3,7} (D)$ and, therefore, $\mathcal{W}_{3,6}(D - 5) = \mathcal{W}_{3,7}(D - 5) = \varnothing$. On the other hand, the reduced network $D - 5$ remains weakly connected.
\\
Finally, $6 \in M_{1,7}(D) \cap M_{2,7}(D) \cap M_{3,7}(D) \cap M_{4,7}(D) \cap M_{5,7}(D)$. Indeed, node 6 is the sole broker of all interaction with node 7. In particular, the network $D - 6$ is not weakly connected and consists of two components: $\{ 1,2,3,4,5 \}$ and $\{ 7 \}$.
\\
All other intermediaries are non-middlemen. Indeed, even the removal of either node 3 or 4 does not affect the connectivity in the network.\footnote{It is worth noting that if all arcs were reciprocated in network $D$ from Figure~\ref{weakmm} to form an undirected network, nodes 2 and 5 would no longer be middlemen. However, node 6 would still be a middleman.}
\end{example}

\noindent
Theorem~\ref{undirectedmiddlemen} below naturally follows from Definition~\ref{strongweakmiddlemen} and the properties illustrated in Example~\ref{identifyingmiddlemen}.

\begin{theorem} \label{undirectedmiddlemen}
Every middleman in a strongly connected network is a strong middleman.
\end{theorem}

\begin{proof}
Consider a strongly connected network $(N,D)$ where $\# N =n \geqslant 3$. According to Definition~\ref{middleman}, a node $h \in N$ is a middleman if it rests on all paths between at least two other nodes, say $i$ and $j$. Since $W_{ij}(D) = W_{ji}(D)$, the property that $h \in \cap \mathcal{W}_{ij}(D)$ implies that $h \in \cap \mathcal{W}_{ji}(D)$.
\\
Thus, in $D - h$ all paths from node $j$ to node $i$ as well as from node $i$ to node $j$ are disconnected. This in turn implies that $i$ and $j$ cannot be connected through any path in any direction and $D - h$ must contain at least two components, separating $i$ and $j$ in different components. This implies that $h$ is actually a strong middleman in $D$.
\end{proof}

\medskip\noindent
The fact that every weakly connected undirected network $D$ on node set $N$ is actually strongly connected gives rise to the following corollary to Theorem \ref{undirectedmiddlemen}:

\begin{corollary}
	Every middleman in an undirected network is a strong middleman.
\end{corollary}

\noindent Regular middlemen only exist due to the distinction between weakly connected and strongly connected nodes. The distinction collapses in an equivalent undirected network as all nodes are effectively strongly connected. This implies the following.

\begin{corollary} \label{corundirectedmiddleman}
Let $D$ be some network on node set $N$ with $n \geqslant 3$. If there exists at least one regular middleman in $D$, then there exist two distinct nodes $i,j \in N$ with $i \neq j$ such that $\mathcal W_{ij} (D) \neq \varnothing$ and $\mathcal W_{ji} (D) = \varnothing$.
\end{corollary}

\noindent
The distinction between regular and strong middlemen is natural and enhances our understanding of the functionality of directed versus undirected networks. We enhance this understanding further in the following discussion that allows the measurement of middleman control, in which it is shown that regular middlemen can actually be more powerful than strong middlemen.

\subsection{Contestation of intermediaries in networks}

Next we examine the relationship between critical nodes and competition or ``contestation'' in networks. Based on the model of network competition in \citet{GillesDiamantaris2013}, such contestation rests on the ability to use alternative pathways to reach other nodes in the network and circumvent a particular intermediary. Hence, it refers to the ability to prevent that intermediary from mediating the information flow to other nodes.

 A node is contested by other nodes if this group of nodes can cover all connections facilitated by that node. Formally, contestation is modelled as the ability of an alternative group of nodes to service the \emph{coverage} of an intermediary, given by the product of that intermediary's predecessor and successor set.

\begin{definition} \label{Contested}
Let $D$ be a network on node set $N=\{1, \ldots ,n\}$ and let $i \in M_D$ be some intermediary node in $D$.
\begin{abet}
\item The \textbf{coverage} of intermediary $i$ in the network $D$ is defined as all node pairs $(h,j)$ with $h \neq j$ that can use node $i$ as an intermediary in their interaction:
\begin{equation}
	\Gamma_i (D) = \left\{ \left. \, (h,j) \in P_i (D) \times S_i (D) \, \right| \, h \neq j \, \right\}
\end{equation}
and the \textbf{extended coverage} of node $i$ in network $D$ is defined by
\begin{equation}
	\overline{\Gamma}_i (D) = \overline{P}_i (D) \times \overline{S}_i (D)
\end{equation}
\item Intermediary $i$ is \textbf{contested} by node set $C \subset N$ in the network $D$ if $i \notin C$ and it holds that
\begin{equation} \label{Group Contested}
\Gamma_i (D) \subseteq \bigcup_{j \in C} \overline{\Gamma}_j (D-i)
\end{equation}
The class of all contesting node sets of intermediary $i$ is denoted by $\mathcal{C}_i (D) \subset 2^N$.
\\
A \textbf{minimal} contesting node set is given by $C_{i}^{*}(D) \in \arg \min \{ \# C \mid C \in \mathcal{C}_{i} (D) \, \}$.

\item Intermediary $i$ is \textbf{directly contested} by a node $j \neq i$ in network $D$ if the singleton node set $\{ j \}$ contests $i$ in $D$, i.e., $\Gamma_i (D) \subseteq \overline{\Gamma}_j (D-i)$.

\item Intermediary $i$ is \textbf{uncontested} if there is no node set that contests $i$.
\end{abet}
\end{definition}

\noindent
The next proposition states in essence the very nature of contestation in a network is that a node can completely take over the functionality of the contested intermediary. Thus, intermediary $i$ is directly contested by node $j$ only when all of $i$'s predecessor set can be connected to $i$'s successor set either through or from node $j$ when $i$ is removed from the network.

\begin{proposition}
	An intermediary node $i \in M_D$ is directly contested by node $j \in N$ in network $D$ if and only if
	\begin{equation} \label{Directly Contested}
	    P_{i}(D) \subseteq P_{j} (D - i) \cup \{j\} \quad \mbox{as well as} \quad S_{i} (D) \subseteq S_{j} (D - i) \cup \{ j \} \, .
	\end{equation}
\end{proposition}

\noindent
The exact same intuition is used with respect to contestation by a group of nodes as shown in the next example.

\begin{example} \label{Simple Contestability}
We consider a network to illustrate the notion of contestability. Consider directed network $D$ on node set $N = \{1,2,3,4,5,6,7\}$ shown in Figure~\ref{weakmm} on page \pageref{weakmm}. Table 1 below provides the predecessor and successor sets of all nodes in the network.

\begin{table}[h]
\begin{center}
\begin{tabular}{|c|cc|}
\toprule
Node & Predecessor Set  & Successor Set                     \\
\midrule
1    & $P_{1}(D)=\varnothing$          & $S_{1}(D)=\{2,3,4,5,6,7\}$        \\
2    & $P_{2}(D)=\{1\}$                & $S_{2}(D)=\{4,5,6,7\}$            \\
3    & $P_{3}(D)=\{1\}$                & $S_{3}(D)=\{5,6,7\}$              \\
4    & $P_{4}(D)=\{1,2\}$              & $S_{4}(D)=\{6,7\}$                \\
5    & $P_{5}(D)=\{1,2,3\}$            & $S_{5}(D)=\{6,7\}$                \\
6    & $P_{6}(D)=\{1,2,3,4,5\}$        & $S_{6}(D)=\{7\}$                  \\
7    & $P_{7}(D)=\{1,2,3,4,5,6\}$      & $S_{7}(D)=\varnothing$  \\
\bottomrule
\end{tabular}

\caption{Predecessor and successor sets of nodes in Figure~\ref{weakmm}}
\label{network1stats}
\end{center}
\end{table}

\noindent
Using this information we deduce that intermediaries 3 and 4 are contested, whereas intermediaries 2, 5, and 6 are uncontested.
\\
Here, node 3 is directly contested by node 2: Indeed, $P_{3}(D) = \{1\} \equiv P_{2}(D)$ and $S_{3}(D) = \{5,6,7\} \subset S_{2}(D) = \{4,5,6,7\}$. It is also true that $P_{3}(D) \subseteq P_{2}(D - 3) \cup \{2\} = \{ 1,2 \}$ and $S_{3}(D) \subseteq S_{2}(D - 3) \cup \{2\} = \{ 2,4,5,6,7 \}$.
\\
This case introduces what can be denoted as \textit{asymmetric contestation}, meaning that although node $i$ contests node $j$, it may not be true that node $j$ contests node $i$. Here, node 3 directly contests node 2, although node 2 is not directly contested by node 3. Only in rare cases will there exist \textit{symmetric contestation} where node $i$ contests node $j$ and node $j$ contests node $i$.
\end{example}

\noindent
The next example highlights more complex forms of contestation where a highly connected node contests two others, while these two nodes in turn contest the highly connected node.

\begin{example} \label{Group Contestability}
Consider directed network $D$ on node set $N = \{1,2,3,4,5,6\}$, shown in Figure~\ref{Complex Contestability}, where $M(D) = \varnothing$. Here, node $4$ connects nodes $1$ and $2$ to nodes $5$ and $6$, and therefore directly contests node $2$ while not being directly contested by any other individual node. However, node $4$ is not a middleman and indeed the set $ C= \{ 2,3 \}$ contests node $4$.

\begin{figure}[h]
\begin{center}
\begin{tikzpicture}[scale=0.5]
\draw[thick, ->] (5,13) -- (9.5,5.6);
\draw[thick, ->] (5,13) -- (5,2.9);
\draw[thick, ->] (10,10) -- (0.8,10);
\draw[thick, ->] (10,10) -- (0.8,5.5);
\draw[thick, ->] (10,10) -- (5.6,2.6);
\draw[thick, ->] (10,5) -- (0.8,5);
\draw[thick, ->] (10,5) -- (0.8,9.7);
\draw[thick, ->] (5,2) -- (0.7,4.5);
\draw[thick, ->] (5,2) -- (0.5,9.3);

\draw (5,13) node[circle,fill=black!10] {$1$};
\draw (10,10) node[circle,fill=black!25] {$2$};
\draw (10,5) node[circle,fill=black!25] {$3$};
\draw (5,2) node[circle,fill=black!25,draw,very thick] {$4$};
\draw (0,5) node[circle,fill=black!10] {$5$};
\draw (0,10) node[circle,fill=black!10] {$6$};

\end{tikzpicture}
\caption{Network $D$ in Example \ref{Group Contestability} where $C = \{ 2,3 \}$ contests node $4$}
\label{Complex Contestability}
\end{center}
\end{figure}
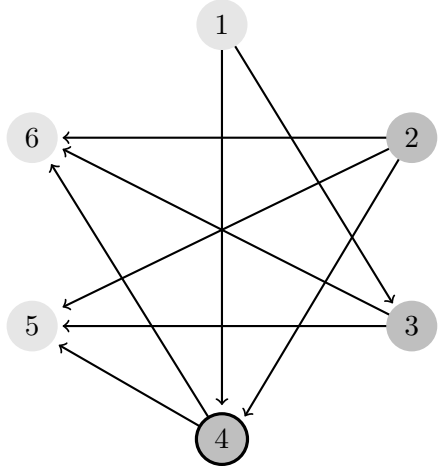

\noindent
Clearly, the extended coverage of nodes $2$ and $3$ encapsulates the coverage of node $4$. Therefore, although nodes $2$ and $3$ do not contest node $4$ individually, the node set $C = \{ 2,3 \}$ contests node $4$. Indeed, the condition for group contestation holds:
\begin{equation}
P_{4}(D) \times S_{4}(D) \subset \left( \, \overline{P}_{2}(D - 4) \times \overline{S}_{2}(D - 4) \cup \overline{P}_{3}(D - 4) \times \overline{S}_{3}(D - 4) \, \right).
\end{equation}
If node $4$ is removed from the network, its function can be fully replaced by the combination of nodes $2$ and $3$ and therefore all other nodes that were connected can still be connected in the same way.
\end{example}

\noindent
Example~\ref{Group Contestability} highlights the requirement for the extended coverage $\overline{\Gamma}_{j}(D - i)$ used in the definition of contestation instead of the coverage, $\Gamma_{j}(D - i)$. Indeed, consider the network in Figure~\ref{Complex Contestability}. As noted, there exist no middlemen and all intermediaries are contested given the definition above. With a more restricted conception based on $\Gamma_{j}(D - i)$, node $4$ would neither be contested nor a middleman. Definition~\ref{Contested} adjusts for predecessors of the given node, $i$, that can connect to the successors of $i$, thereby fulfilling the same function and, thus, contesting $i$.

Examples~\ref{Simple Contestability} and~\ref{Group Contestability} give an indication that if an intermediary is contested, it cannot be a middleman. For example, in Figure~\ref{weakmm} agent $3$ is a non-middleman because his function is directly contested by the presence of node $2$, and node $2$ is a middleman because its function is not contested by any other node in the network. Our main result states a duality between contestation and the existence of middlemen.
\begin{theorem} \label{duality} \textbf{\emph{(Duality of middlemen and contestability)}} \\
Consider a network $D$ on node set $N$ with $n \geqslant 3$. Then:
\begin{abet}
	\item Every middleman $i \in \mathbf M (D)$ is an uncontested intermediary in $D$.
	\item If an intermediary $i \in M_D$ is uncontested in $D$, then $i$ is a middleman, i.e., $i \in \mathbf M (D)$.
\end{abet}
\end{theorem}
\begin{proof}
Let $D$ be a network on node set $N$ with $n \geqslant 3$.

\smallskip\noindent
\emph{Proof of (a):}
The condition for contestability stated in equation (\ref{Group Contested}) on page \pageref{Group Contested} contends that a node $h \in N$ is contested in network $D$ if its coverage, determined by the nodes it intermediates, is a subset of the coverage and the reach of the nodes in $C_{h} \subset N \setminus \{ h \}$.
\\
Now consider an intermediary $i \in M_D$ that is contested by a set of agents, $C \subset N \setminus \{ i \}$. Since $i$ is contested, it must be true that all of $i$'s predecessors can be connected to all of $i$'s successors by a path that does not include $i$. Therefore, $i$ cannot be a middleman. This implies the assertion that every middleman is uncontested.

\smallskip\noindent
\emph{Proof of (b):}
Consider an intermediary $h \in M_D$ who is uncontested in the network $D$. Then $h$'s coverage is not a subset of the coverage of any set of nodes plus the respective reach of each of these nodes. This implies that $h$ itself has to rest on at least one path that no other nodes in the network rest on when $h$ is removed from the network. Hence, in the network $D$ there exists at least one pair of nodes, say $i$ to $j$, with $h \in \cap \mathcal{W}_{ij} (D)$ and $\mathcal{W}_{ij}(D - \{h\}) = \varnothing$. This implies that $h$ is actually a middleman concerning the paths from $i$ to $j$.
\end{proof}

\medskip\noindent
From Theorem \ref{duality}, all middlemen are uncontested; if a node is contested, then all of its intermediation functions or coverage can be replaced by the coverage of other nodes. From this, it is understood that a middleman is an intermediary that has a unique function and is, in some way, more effective than non-middlemen with respect to their connectivity and thus coverage in the network.

\section{Measuring middleman power}
\label{sec:middlemanPower}

A middleman occupies a critical position in a network since its removal disconnects at least two or more other nodes in the network and, in the most extreme case, might separate the network into multiple components. Therefore, it seems logical to ask how we can measure this power. After examining established measures, we propose a measure of middleman power based on the disconnections that emerge when a middleman is removed from the network.

\subsection{Middlemen and betweenness centrality}

We first examine whether betweenness centrality could be a tool to assess middleman power. \emph{Betweenness centrality} was proposed independently by \citet{Anthonisse1971} for edges and rephrased by \citet{Freeman1977} for nodes in undirected networks. \citet{White1994} proposed an extension to directed networks.

This measure seems specifically relevant since it explicitly considers the role of a node in connecting other nodes in the network. It may be expected that the betweenness centrality score of a node provides an indication of what nodes are middlemen by having a greater betweenness centrality than non-middlemen in the network.

To define the betweenness centrality measure, let $\pi(hj)$ be the number of shortest paths---or \emph{geodesics}---from node $h$ to node $j$. Furthermore, let $\pi_{i}(kj)$ be the number of geodesics that pass through node $i$. Betweenness centrality is now defined as

\begin{equation} \label{betweennesscentrality}
BC_{i}(D) = \sum_{h,j \neq i \colon \pi(hj) \neq 0} \frac{\pi_{i}(hj)}{\pi(hj)}
\end{equation}

Equation (\ref{betweennesscentrality}) indicates that a middleman $i$ between nodes $h$ and $j$ would always have a high betweenness centrality since by Definition~\ref{middleman} a middleman is on all paths between these two nodes. In particular, $\pi_i (hj) = \pi (hj)$. However, the formulation of betweenness centrality $BC$ only considers geodesics. As the next example illustrates, the betweenness centrality of non-middlemen may even surpass that of middlemen.

\begin{figure}[h]
\begin{center}
\begin{tikzpicture}[scale=0.5]
\draw[thick, ->] (0,10) -- (4.3,10);
\draw[thick, ->] (0,0) -- (4.3,0);
\draw[thick, ->] (0,5) -- (4.3,5);

\draw[thick, ->] (5,10) -- (9.2,8);
\draw[thick, ->] (5,10) -- (9.2,3);

\draw[thick, ->] (5,5) -- (9.2,7.3);
\draw[thick, ->] (5,5) -- (9.1,2.5);

\draw[thick, ->] (5,0) -- (9.5,6.8);
\draw[thick, ->] (5,0) -- (9.2,2);

\draw[thick, ->] (10,7.5) -- (14,7.5);
\draw[thick, ->] (10,2.5) -- (14,2.5);

\draw[thick, ->] (10,7.5) -- (14,2.7);
\draw[thick, ->] (10,2.5) -- (14,7.4);

\draw (0,0) node[circle,fill=black!10] {$3$};
\draw (0,5) node[circle,fill=black!10] {$2$};
\draw (0,10) node[circle,fill=black!10] {$1$};
\draw (5,0) node[circle,fill=black!25,draw,thick] {$6$};
\draw (5,5) node[circle,fill=black!25,draw,thick] {$5$};
\draw (5,10) node[circle,fill=black!25,draw,thick] {$4$};
\draw (10,7.5) node[circle,fill=black!10] {$7$};
\draw (10,2.5) node[circle,fill=black!10] {$8$};
\draw (15,7.5) node[circle,fill=black!10] {$9$};
\draw (15,2.5) node[circle,fill=black!10] {$10$};

\end{tikzpicture}
\end{center}
\caption[Differentiating betweenness and middlemen]{The network $D$ considered in Example~\ref{ex:BC}}
\label{mmnmm}
\end{figure}
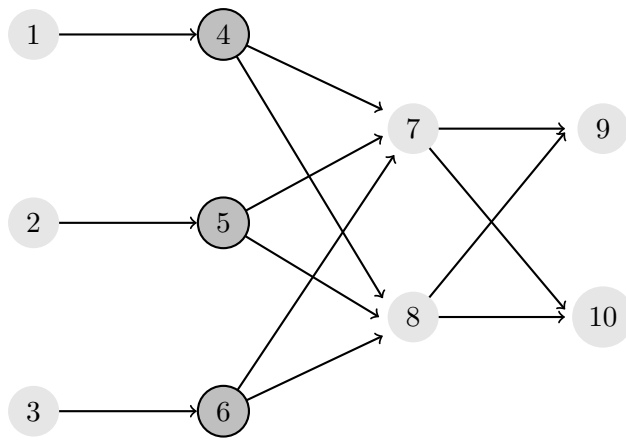

\begin{example} \label{ex:BC}
Consider the acyclic network $D$ depicted in Figure~\ref{mmnmm}. Here, as indicated in the graph, $\mathbf{M} (D) = \{ 4, 5, 6 \}$, while nodes $7$ and $8$ are contested intermediaries. All middlemen have the same non-normalised betweenness centrality measures due to their equivalent positions: $BC_{4}(D) = BC_{5}(D) = BC_{6}(D) = 4$. However, both contested intermediaries have larger non-normalised betweenness centrality scores: $BC_{7}(D)=BC_{8}(D)=6$.
\\
In the underlying undirected network, $U$, where all arcs in $D$ are reciprocated, the non-middlemen still have a higher betweenness centrality than middlemen: $BC_{4}(U)=BC_{5}(U)=BC_{6}(U)=16.4$ and $BC_{7}(U)=BC_{8}(U)=25$. Below, Table 2 provides a comparison of common centrality measures all of which indicate that nodes $7$ and $8$ are more central, therefore underrating the nodes with powerful middleman properties.
\end{example}

\begin{table}[h]
\begin{center}
\begin{tabular}{|c|cccccc|}
\toprule
Node & Degree 	& PageRank	& Betweenness 	& Closeness 	& Bonacich 	& $\beta$-Measure \\
\midrule
1    & 1    	& 0.112 	& 0.000    		& 0.360  		& 0.328 	& 0.333\\
2    & 1    	& 0.112 	& 0.000    		& 0.360  		& 0.328 	& 0.333\\
3    & 1    	& 0.112 	& 0.000    		& 0.360  		& 0.328 	& 0.333\\
4    & 3    	& 0.330 	& 0.456    		& 0.474  		& 1.047 	& 1.400\\
5    & 3    	& 0.330 	& 0.456    		& 0.474  		& 1.047 	& 1.400\\
6    & 3    	& 0.330 	& 0.456    		& 0.474  		& 1.047 	& 1.400\\
7    & 4    	& 0.480 	& 0.694    		& 0.692  		& 1.565 	& 2.000\\
8    & 4    	& 0.480 	& 0.694    		& 0.692  		& 1.565 	& 2.000\\
9    & 2    	& 0.297 	& 0.011    		& 0.474  		& 0.863 	& 0.400\\
10   & 2    	& 0.297 	& 0.011    		& 0.474  		& 0.863 	& 0.400\\
\bottomrule
\end{tabular}
\caption{Centrality results for the undirected network $U$}
\end{center}
\end{table}

\noindent
The identified deficiency of betweenness centrality to verify and measure the control exercised by middlemen, highlighted in Figure \ref{mmnmm} and Table 2, extends to other less common centrality measures. Indeed, no standard measure for undirected networks specifically identifies and highlights middlemen, instead the measures tend to over-inflate the power and importance of contested intermediaries in networks.

The potential high betweenness centrality of non-middlemen, for example, follows from the underlying assumptions of the measure. Indeed, only geodesics are counted between two given nodes, assuming that these geodesic paths have equal weight. The combination of these assumptions implies that the betweenness centrality measure does not necessarily measure the ``power'' of a node in negotiating between two others.


\subsection{A middleman power measure}

The power or control of a middleman in a network $D$ can be measured by simply counting the number of node pairs $(i,j) \in N \times N$ such that $(i,j)$ is connected in $D$, while $(i,j)$ is not connected in $D-i$. This number is introduced as the \emph{brokerage} of that middleman.

We introduce a general counting method that identifies the brokerage of an arbitrary node, rather than just the middlemen in a network. Therefore, this counting method introduces a way to exactly identify the middlemen in the network in an algorithmic fashion.

\begin{definition}
	Let $D$ be a network on node set $N = \{1, \ldots ,n\}$ and let $i \in N$ be an arbitrary node. The \textbf{brokerage} of node $i$ is
	\begin{equation} \label{brokerage}
		b_{i}(D) = \sum_{j \neq i} \left[ \, \# S_{j}(D) - \# S_{j}(D - i) \, \right] - \#P_{i}(D) .
	\end{equation}
\end{definition}

\noindent
The successor set of a node contains all other nodes that can be reached by a path from that node. The first part of (\ref{brokerage}), $\sum_{j \neq i} \# S_{j}(D)$, counts the total number of successors of all $n$ nodes except for node $i$. Hence, it provides an indication of the total connectivity of the network as a whole except for the connections of node $i$ in $D$.

Using the same intuition, the second part of (\ref{brokerage}), $\sum_{j \neq i} \# S_{j}(D - i)$, refers to the total connectivity of the network when node $i$ is removed. We remark that $\sum_{j \neq i} \# S_{j}(D) > \sum_{j \neq i} \# S_{j}(D - i)$ if $d_{i}(D) \geqslant 1$, and $\sum_{j \neq i} \# S_{j}(D) = \sum_{j \neq i} \# S_{j}(D - i)$ if $d_{i} (D) = 0$.

Therefore, $\sum_{j \neq i} \left[ \, \# S_{j}(D) - \# S_{j}(D - i) \, \right]$ expresses the \emph{net connectivity differential} from the removal of node $i$ from the network $D$. The net connectivity differential captures two features: (1) The direct connectivity of node $i$ in terms of its successor and predecessor set;\footnote{Indeed, the larger the predecessor and successor sets of node $i$ the larger the differential will be regardless of whether $i$ is a middleman or not.} and (2) The lost connectivity to other nodes not including $i$.

To assess the impact of a middleman, we are only interested in the lost connectivity caused by its removal from the network. Therefore, we compensate the connectivity differential with the upward connectivity of $i$ in the network. Specifically, the predecessor set of node $i$ has to be removed from the connectivity differential, thus adjusting for the direct connectivity of node $i$, resulting in the formulation in (\ref{brokerage}).

In short, the brokerage of a node counts the number of third-party disconnections that occur due to the removal of a node; or in other words, counts the number of $ij$-middleman sets that a node is a member of.

If $b_{i}(D) = 0$ then the removal of $i$ from the network makes no change to the network's connectivity---compensating for the connectivity of $i$. Hence, all nodes that are connected by a path in $D$ can still be connected in $D - i$. Thus, compensating for their connection to $i$ in $D$, the number of successors of all $j$ nodes is the same in $D - i$ as in $D$.

On the other hand, if $b_{i}(D) \geqslant 1$, there exists at least one pair of connected nodes that are now not connected in $D - i$. As a consequence, $i$ must be a middleman.

\medskip\noindent
We normalise the brokerage of a node by calculating the total number of potential opportunities for brokerage in the network. Brokerage---and, therefore, middleman positions---can only emerge if a pair of nodes are a minimum distance of two or more away from each other. Intuitively, by calculating the indirect successors of all nodes in the network, the total number of brokerage opportunities can be derived.

The set of indirect successors of $i$ in $D$ is given by $S_{i}(D) \setminus s_{i}(D)$. Therefore, the number indirect successors for node $i$ is given by $\# S_{i}(D) - \# s_{i}(D)$. Given this, the maximal potential brokerage in $D$ is computed as:
\begin{equation} \label{normalisation}
B(D) = \sum_{i \in N} \left[ \# S_{i}(D) - \# s_{i}(D) \right] .
\end{equation}
Note that $B(D)=0$ for certain networks, including the empty and complete networks on $N$. This leads to the following definition of a normalised brokerage-based centrality measure.

\begin{definition} \label{middlemanpower}
Let $D$ be a network on node set $N$. The \textbf{middleman power measure} for $D$ is the function $\nu (D) \colon N \to \mathbb{R}_+$ with for every node $i \in N \colon$
\begin{equation} \label{mmpowerindex}
\nu_{i}(D) = \frac{b_{i}(D)}{ \max \{ B(D) , 1 \}} .
\end{equation}
\end{definition}

\noindent
A middleman has a network power of 1 if it brokers all potential opportunities in the network. This includes nodes at the centre of star networks. The next example explicitly computes the middleman power measure for an undirected star and a directed cycle.

\begin{example} \label{starcycle}
For an \emph{undirected} star network, $D^{\star}$, on node set $N = \{ 1, \ldots ,n \}$, where $n \geqslant 3$, it holds that $b_{i}(D^{\star}) = (n-1)(n-2)$ for the centre node and $b_{j}(D^{\star})=0$ for all other nodes. The potential total brokerage for an undirected star is computed as $B(D^{\star}) = (n-1)(n-2)$. Therefore, the middleman power of the centre node is
\begin{equation}
\nu_{i}(D^{\star}) = \frac{(n-1)(n-2)}{(n-1)(n-2)} = 1 .
\end{equation}

\noindent Next, consider a directed cycle $D^{\circ}$ on node set $N$. Each node has an in-degree of 1 and an out-degree of 1, implying that all nodes are intermediaries as well as middlemen. We now compute that $b_{1}(D^{\circ}) = \ldots = b_{n}(D^{\circ}) = \frac{(n-1)(n-2)}{2}$. The potential total brokerage is $B (D^{\circ}) = n(n-2)$, implying $\nu_i = \frac{n-1}{2n}$ for all $i \in N$ where $n \geqslant 3$.
\end{example}

\noindent
We derive several properties for the middleman power measure stated in (\ref{mmpowerindex}).

\begin{theorem} \label{middlemanpowert}
Let $D$ be a network on node set $N=\{1, \ldots ,n\}$.
\begin{abet}
\item For every node $i \in N \colon 0 \leqslant \nu_{i} (D) \leqslant 1$.
\item For every contested intermediary $i \in M_D \colon \nu_{i}(D) = 0$.
\item For every middleman $i \in \mathbf M (D) \colon \nu_i (D) >0$.
\end{abet}
\end{theorem}

\begin{proof}
We show the two assertions subsequently.
\\[1ex]
\emph{Proof of (a):}
We omit a mathematical proof of (a), due to its tedious nature. Instead, we provide an intuitive, more descriptive reasoning.
\\
A middleman cannot take advantage of more than all brokerage opportunities present in a network; therefore $B'(D) \geqslant b_{i}(D)$, implying $\nu_{i}(D) \leqslant 1$\footnote{Only in a network where a middleman rests on all geodesic paths of length two, for example an undirected star, it holds that $B'(D) = b_{i}(D)$.}.
\\
Furthermore, neither $b_{i}(D) < 0$ nor $B(D) < 0$. The minimum brokerage of some node $k$ is in an empty network where $\# S_{k}(D) = \# P_{k}(D) = 0$. In that case, $\sum_{i \in N} \# S_{i}(D) = \sum_{i \neq k} \# S_{i}(D - \{k\})$ since $k$ has no connectivity in the network. Therefore, $\nu_i (D) \geqslant 0$ for any node $i \in N$.
\\[1ex]
\emph{Proof of (b):}
Let $h \in M_D$ be a contested intermediary in the network $D$. Theorem~\ref{duality} asserts a duality between being a contested intermediary and being a non-middleman. Definition~\ref{middleman} implies that $h \in M_D$ is not a middleman if there is no pair $i,j \in N$ with $i \neq j$ such that $h$ lies on all paths from $i$ to $j$ in $D$. Hence, $h \notin \cap \mathcal{W}_{ij}(D)$ for all $i,j \in N$ with $i \neq j$.
\\
Since $h$ is an intermediary it holds that $\# P_{h}(D) > 0$, $\# S_{h}(D) > 0$ as well as $\sum_{i \in N} \# S_{i}(D) > \sum_{i \in N \setminus \{h\}} \# S_{i}(D - h)$ since the nodes in $D - h$ can obviously not connect to $h$. Also, since the connectivity of the network $D$ is not affected by the removal of node $h$, it holds that the removal of $h$ only affects the connectivity with $h$ itself. Hence, $\sum_{i \in N} \# S_{i}(D) - \sum_{i \in N \setminus \{ h \}} \# S_{i}(D - h) = \# S_{h}(D) + \# P_{h}(D)$. Therefore, $b_h (D) =0$, implying that $\nu_{h} (D) = 0$.
\\[1ex]
\emph{Proof of (c):}
Let $i \in \mathbf M (D)$ be a middleman in the network $D$. Then by Theorem \ref{duality}, node $i$ is uncontested and, therefore, assertion (b) does not apply to $i$.
\\
Since $i$ is a middleman, there exist at least two nodes $h,j \in N \setminus \{ i \}$ with $h \neq j$ such that $i$ is an $hj$-middleman. Hence, $j \in S_h (D)$ and $j \notin S_h (D-i)$. Furthermore, since $i$ is a $hj$-middleman it also holds that $j \in P_i (D) \setminus S_i (D)$ and $h \in S_i (D) \setminus P_i(D)$. This implies that
\begin{align*}
	b_i(D) + \#P_{i}(D) & =  \sum_{i' \neq i} \left[ \, \# S_{i'}(D) - \# S_{i'}(D - i) \, \right] \geqslant \\
	& \geqslant \# S_{j}(D) - \# S_{j}(D - i) \geqslant \\
	& \geqslant \#P_{i}(D) +1.
\end{align*}
Hence, $b_i(D) \geqslant 1$, showing the assertion.
\end{proof}

\noindent
Theorem \ref{middlemanpowert} states that the middleman measure $\nu (D) \colon N \to [0,1]$ indeed only assigns a non-zero value to the middlemen in $D$. It ranks the middlemen in $D$ according to the number of flows that are controlled by each of these middlemen. The higher the middleman measure of a middleman $i \in \mathbf M (D)$, the more control that middleman exercises in the network.

\citet{BlochJackson2016} show that, although prominent centrality measures in network analysis make use of different information about each node's position in the network, these measures all originate from a common set of principles that are characterised by the same simple axioms. In particular, these standard measures are all based on a monotonic and additively separable treatment of a network statistic that captures a node's position in the network.

Our middleman measure is not subject to the analysis introduced by \citet{BlochJackson2016}, since it is not founded on any of the network statistics identified there. It forms a truly alternative way to assess the importance of a node in a network. Applications show that the middleman measure indeed identifies nodes with the highest impact on the functionality of the network. (See also Section 5 of this paper.)

\subsection{The robustness of middlemen in networks}

The control that middlemen exert in a network can be affected by deliberate actions of other parties in that network. Indeed, other nodes can create new relationships in the network and sever arcs to counter the power and control exerted by a middleman. This refers to a \emph{dynamic} element, represented as a modification of the architecture or topology of the network to render a middleman to a non-middleman position. Therefore, a more robust middleman position is less susceptible to a change in the topological structure of the network.

We refer to the number of changes required in the network to render a middleman powerless as the \emph{robustness} of that middleman. We perceive middleman robustness to infer how a given middleman can maintain an exploitive position given a change to the topological structure of the network from the addition or deletion of arcs.

We introduce three methods in which to measure middleman robustness: The first two measures relate the robustness of a node's exploitive position given the deletion and addition of arcs; The third method measures middleman robustness given the removal of all arcs from and to a certain set of nodes from the initial network. We consider the first arc robustness measure as the most essential one, referring to it as the ``robustness'' of a middleman.

\paragraph{Arc robustness.}

The \emph{arc-robustness}---or simply \emph{robustness}---of a middleman in a network is defined as the minimum number of arcs that have to be \emph{added} to the network in order for a middleman to be rendered inessential for maintaining all flows in the network.

The dual formulation of robustness---denoted as  \emph{dual robustness}---measures the minimum number of arcs that have to be \emph{removed} from the network such that a given middleman loses its brokerage function and power.

\begin{definition} \label{robustness}
Let $D$ be some network on node set $N = \{1,\ldots,n\}$ such that $\mathbf{M}(D) \neq \varnothing$. Furthermore, let $i \in \mathbf{M}(D)$ be a middleman in $D$.
\begin{abet}
	\item The \textbf{robustness} of middleman $i$ is given by
	\begin{equation}
		\rho_{i} (D) = \min \, \left\{ \, \# D' \mid D \subset D' \mbox{ and } i \notin \mathbf{M}(D') \, \right\} - \# D
	\end{equation}
	
	\item The \textbf{dual robustness} of middleman $i$ is given by
	\begin{equation}
		\rho_{i}^{\star} (D) = \# D - \max \, \left\{ \# D' \mid D' \subset D \mbox{ and } i \notin \mathbf{M}(D') \, \right\}
	\end{equation}
\end{abet}
\end{definition}

\noindent
Obviously, the cardinality of $D$ is equal to the number of arcs in the network. The arc-robustness of a middleman is defined as the minimum number of arcs that need to be added to the initial network $D$ in order for a middleman to be completely circumvented and subsequently lose its position as a middleman. The dual-robustness measure is equal to the minimum number of arcs that need to be removed from the initial network $D$ such that a given middleman no longer has an exploitive position in the resulting network $D' \subset D$.

If a node is a middleman then the value of both its arc-robustness and its dual-robustness will always be a positive integer. However, neither of these robustness measures will necessarily be correlated with the middleman's brokerage measure; there exist situations in which a middleman's brokerage can be increased by extending the network beyond that middleman, while keeping the middleman's robustness and dual-robustness constant.

An alternative interpretation of the robustness of a middleman is how easily this middleman can become contested by a coalition of other nodes. A low value for its robustness $\rho_{i}$ implies that it takes relatively few new relationships to make that middleman, $i$, contestable.

The following result gives bounds on these two arc robustness measures for an arbitrary middleman.
\begin{proposition} \label{prop:bounds}
	Let $D$ be a network on the node set $N$ and let $i \in \mathbf M (D)$ be a middleman in $D$. Then it holds that
	\begin{gather}
		1 \leqslant \rho_i (D) \leqslant \min \, \{ b_i (D) , d^+_i + d^-_i -1 \} \\
		1 \leqslant \rho^{\star}_i (D) \leqslant \min \, \{ d^+_i , d^-_i \}
	\end{gather}
\end{proposition}

\noindent
The lower bounds stated in Proposition \ref{prop:bounds} are obvious. Indeed, a middleman is critical for the connection for at least one pair of nodes. So, robustness requires the introduction of at least one new arc to connect such a pair. Similarly, dual robustness requires the removal of at least one arc to make a middleman obsolete. This lower bound is attained for both robustness measures in a circular directed network in which all nodes are middlemen and each node can be circumvented by an arc from its direct predecessor to its direct successor---or by the removal of the single incoming or outgoing arc of that middleman.

The proof of the stated upper bound for the robustness measure $\rho$ in Proposition \ref{prop:bounds} is based on some simple insights. First, it should be clear that in every network, by linking the node pairs that are counted in the brokerage measure $b_i$ of some middleman, one indeed circumvents that middleman completely. So, the robustness of a middleman never exceeds the number of pairs that are counted in the brokerage of that middleman.

Furthermore, the example also shows that the second identified upper bound $d^+_i + d^-_i -1$ is attained exactly in a star network. The methodology of circumventing a middleman as constructed in the example, applies in general to any middleman in any network, thus showing that this upper bound indeed applies generally.

Finally, the upper bound for the dual robustness measure is identified from the next example as well. Indeed, in a star network either all incoming arcs or all outgoing arcs need to be removed to make a middleman a source or a sink in the resulting modified network. This methodology again applies generally to any middleman in any network, thus providing an upper bound. As shown in the example, this upper bound is attained in a star network.

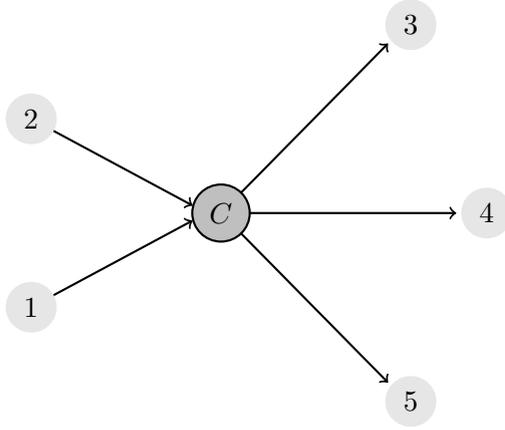
\begin{figure}[h]
\begin{center}
\begin{tikzpicture}[scale=0.5]
\draw[thick, ->] (0,2.5) -- (4.25,4.8);
\draw[thick, ->] (0,7.5) -- (4.25,5.2);

\draw[thick, ->] (5,5) -- (9.4,9.5);
\draw[thick, ->] (5,5) -- (11.2,5);
\draw[thick, ->] (5,5) -- (9.4,0.5);

\draw (0,2.5) node[circle,fill=black!10] {$1$};
\draw (0,7.5) node[circle,fill=black!10] {$2$};
\draw (5,5) node[circle,fill=black!25,draw,thick] {$C$};
\draw (10,10) node[circle,fill=black!10] {$3$};
\draw (12,5) node[circle,fill=black!10] {$4$};
\draw (10,0) node[circle,fill=black!10] {$5$};

\end{tikzpicture}
\end{center}
\caption{A star network $D$ in Example \ref{ex:Star}.}
\label{fig:Star}
\end{figure}

\begin{example} \label{ex:Star}
	Consider the network $D$ depicted in Figure \ref{fig:Star} on the node set $N = \{ 1, \ldots ,5,C \}$. Clearly, $\mathbf M (D) = \{ C \}$ is the unique middleman in this network with degrees $d^-_C =2$ and $d^+_C =3$. The brokerage of this middleman is determined as $b_C (D) = d^-_C \cdot d^+_C = 2 \cdot 3 =6$.
	\\
	The middleman $C$ can be rendered powerless by introducing a minimum of 4 new arcs into the network $D$ to form network $D'$, depicted in Figure \ref{fig:StarMod}. This shows that, indeed, the robustness of middleman $C$ in the network $D$ is exactly the indicated upper bound in Proposition \ref{prop:bounds}: $\rho_C (D) = d^-_C + d^+_C -1 =4$, which is based on the length of the semi-circular path around the middleman $C$ to connect the direct predecessors and direct successors of $C$.
	\\
	Finally, we also note that middleman $C$ in network $D$ can be rendered powerless by making $C$ a source by removing the arcs from nodes 1 and 2 to $C$. This indeed shows that in a star network the dual robustness of the middleman $C$ is indeed equal to the indicated upper bound in Proposition \ref{prop:bounds}: $\rho_C^{\star} (D) = \min \, \{ d^+_C , d^-_C \} = \min \{ 2,3 \} =2$.
\end{example}

\begin{figure}[h]
\begin{center}
\begin{tikzpicture}[scale=0.5]
\draw[thick, ->] (0,2.5) -- (4.25,4.8);
\draw[thick, ->] (0,7.5) -- (4.25,5.2);

\draw[thick, ->] (5,5) -- (9.4,9.5);
\draw[thick, ->] (5,5) -- (11.2,5);
\draw[thick, ->] (5,5) -- (9.4,0.5);

\draw[dashed,thick, ->] (0,2.5) -- (0,6.75);
\draw[dashed,thick, ->] (0,7.5) -- (9.2,10);
\draw[dashed,thick, ->] (10,10) -- (12,5.75);
\draw[dashed,thick, ->] (12,5) -- (10,0.75);

\draw (0,2.5) node[circle,fill=black!10] {$1$};
\draw (0,7.5) node[circle,fill=black!10] {$2$};
\draw (5,5) node[circle,fill=black!10] {$C$};
\draw (10,10) node[circle,fill=black!10] {$3$};
\draw (12,5) node[circle,fill=black!10] {$4$};
\draw (10,0) node[circle,fill=black!10] {$5$};

\end{tikzpicture}
\end{center}
\caption{The modified star network $D'$ in Example \ref{ex:Star}.}
\label{fig:StarMod}
\end{figure}
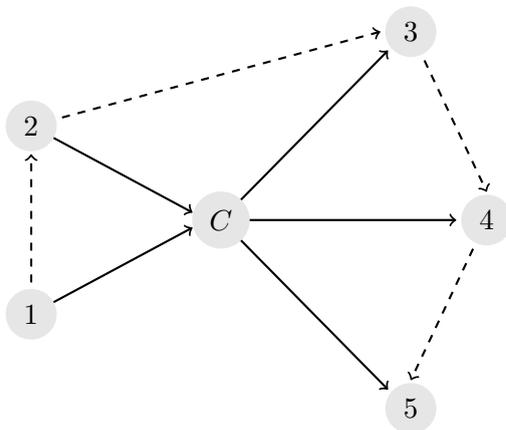

\paragraph{Node robustness.}

The \emph{node-robustness} of a middleman is an extension of the dual (arc) robustness measure. The node-robustness of a middleman is defined in terms of the minimum number of nodes that need to be deleted from the network in order for a given node to lose its middleman position. The removal of a node $i \in N$ from network $D$ is equivalent to the deletion of all arcs to its direct neighbours, $j \in s_i (D) \cup p_i(D)$.

\begin{definition}
Let $D$ be a network on node set $N$ such that $\mathbf{M}(D) \neq \varnothing$. Furthermore, let $i \in  \mathbf{M}(D)$ be a middleman in $D$. The \textbf{node--robustness} of middleman $i$ is now given as
\begin{equation}
\psi_{i} (D) = \min \, \{ \, \# C \mid C \subset N \setminus \{ i \} \mbox{ such that } i \notin \mathbf{M}(D-C) \, \}
\end{equation}
where $D-C = \{ ij \in D \mid i,j \in N \setminus C \, \}$ is the network resulting from $D$ through the removal of the node set $C$.
\end{definition}

\noindent
It can easily be checked that the node-robustness of a middleman is at most its dual robustness. Indeed, every arc that is deleted to render the middleman powerless, can be replaced by the removal of the originating node of each removed arc.
\begin{proposition} \label{prop:NodeBounds}
	Let $D$ be a network on the node set $N$ and let $i \in \mathbf M (D)$ be a middleman in $D$. Then it holds that
	\begin{equation}
		1 \leqslant \psi_i (D) \leqslant \rho^{\star}_i (D) \leqslant \min \, \{ d^+_i , d^-_i \} .
	\end{equation}
\end{proposition}

\noindent
The identified bounds in Proposition \ref{prop:NodeBounds} can be illustrated with a simple line network consisting of three nodes, illustrated in Figure \ref{fig:Line} below, in which the middleman's node robustness is strictly lower than its dual robustness.

In Figure \ref{fig:Line}, node 2 is the unique middleman, while nodes 1 and 3 are leaf nodes. The removal of a single arc retains 2's middleman position. Therefore, $\rho^{\star}_2 =2 = d^+_2=d^-_2$. On the other hand, the node robustness of node 2 is determined by the fact that the removal of either node 1 or node 3 is sufficient to render node 2 no longer being a middleman. Hence, $\psi_2 = 1 < \rho^{\star}_2$.

\begin{figure}[h]
\begin{center}
\begin{tikzpicture}[scale=0.5]
\draw[thick, ->] (0,0.25) -- (4.3,0.25);
\draw[thick, ->] (5,-0.25) -- (0.7,-0.25);

\draw[thick, ->] (5,0.25) -- (9.3,0.25);
\draw[thick, ->] (10,-0.25) -- (5.7,-0.25);

\draw (0,0) node[circle,fill=black!10] {$1$};
\draw (5,0) node[circle,fill=black!25,draw,thick] {$2$};
\draw (10,0) node[circle,fill=black!10] {$3$};

\end{tikzpicture}
\end{center}
\caption{The difference between node robustness and dual link robustness}
\label{fig:Line}
\end{figure}
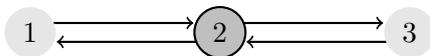

Furthermore, we point out that there are numerous networks in which the upper bound identified in Proposition \ref{prop:NodeBounds} is tight. Indeed, in the star network depicted in Figure \ref{fig:Star}, the middleman $C$ has equal dual link robustness and node robustness: $\psi_C = \rho^{\star}_C = \min \{ d^+_C , d^-_C \} = 2$.

\section{An application to two empirical networks}
\label{sec:empiricalNetwork}

We apply our middleman power and robustness measures to two well-known social networks. From the assessment of these networks we provide a discussion regarding the potential of middlemen in these networks. The results of middleman power are compared with other measures of centrality. This is done in terms of reference; we refrain from correlating the results of these measures because we showed above that middleman power measures different aspects of a node than other measures.

\subsection{Middlemen in Krackhardt's advice network}

\begin{figure}[h]
\centering
\includegraphics[scale=0.45]{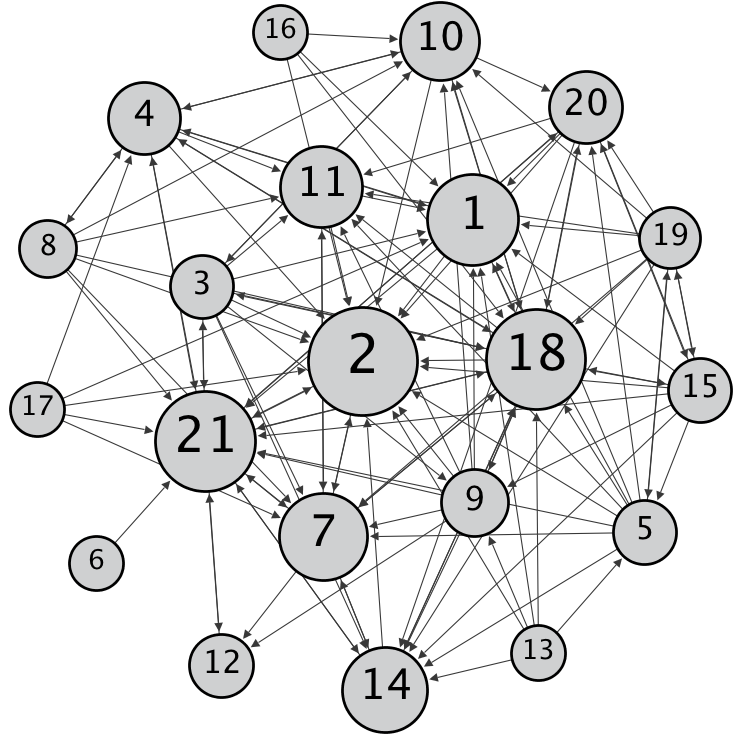}
\caption{Krackhardt's network of advice among managers}
\label{krackhardtnetwork}
\end{figure}

We consider the well-known organisational advice network seminally investigated by \citet{Krackhardt1987}. Krackhardt investigated the relationships between managers in a middle-size firm,\footnote{We use data in the ``LAS'' matrix from p.~129 in the Krackhardt article as it seems to be the most objective measure.} consisting of 21 managers. He collected information from the managers about who sought advice from whom, depicted in Figure~\ref{krackhardtnetwork}. An arc from $i$ to $j$ denotes that manager $i$ has sought advice from manager $j$; therefore, an arc from $j$ to $i$ denotes that manager $j$ has provided advice to manager $i$. In this depiction the size of the node reflects its in-degree.

Middlemen are important for this particular network for a number of intuitive reasons: First, a middleman can block ideas, advice, and information from being transmitted from one group of managers to another. Second, a middleman can manipulate the information transferred from one group of managers to another.

\begin{table}[h]
\begin{center}
\begin{tabular}{|c|ccccc|}
\toprule
Manager       & $d^-_i$ & $d^+_i$& $E_i$    & $BC_i$   & $\nu_i$      \\
\midrule
1             & 12 & 4           & 0.068    & 0.035    & 0.000 \\
2             & 18 & 2           & 0.306    & 0.011    & 0.000 \\
3             & 3 & 9            & 1.271    & 0.018    & 0.000 \\
4*          & 6 & 7            & 1.001    & 0.071    & 0.090 \\
5             & 3 & 10           & 1.463    & 0.009    & 0.000 \\
6             & 0 & 1            & 0.172    & 0.000    & 0.000 \\
7             & 11 & 6           & 0.776    & 0.048    & 0.000 \\
8             & 1 & 7            & 1.013    & 0.001    & 0.000 \\
9             & 4 & 9            & 1.171    & 0.011    & 0.000 \\
10            & 8 & 5            & 0.820    & 0.018    & 0.000 \\
11            & 9 & 3            & 0.344    & 0.004    & 0.000 \\
12            & 3 & 1            & 0.172    & 0.000    & 0.000 \\
13            & 0 & 6            & 0.938    & 0.000    & 0.000 \\
14            & 10 & 4           & 0.625    & 0.002    & 0.000 \\
15*         & 3 & 9            & 1.265    & 0.092    & 0.161 \\
16            & 0 & 4            & 0.580    & 0.000    & 0.000 \\
17            & 0 & 5            & 0.673    & 0.000    & 0.000 \\
18            & 15 & 12          & 1.745    & 0.231    & 0.000 \\
19            & 2 & 10           & 1.493    & 0.002    & 0.000 \\
20            & 6 & 7            & 1.028    & 0.028    & 0.000 \\
21**        & 15 & 8           & 1.348    & 0.176    & 0.147 \\
\bottomrule
\end{tabular}
\caption{Influence, centrality, and middlemen in Krackhardt's advice network}
\label{tabkrackhardt}
\end{center}
\end{table}

Table~\ref{tabkrackhardt} reports the characteristics of this network. Here, we report the in- and out-degrees; the Bonacich centrality index $E$ \citep{Bonacich1972,Bonacich1987}; the betweenness centrality $BC$; and the middleman power measuree $\nu$. Furthermore, a single star (*) indicates a regular middleman and a double star (**) indicates a strong middleman. We identify two regular middlemen, managers 4 and 15, and one strong middleman, manager 21. 

Middleman 15 has the highest middleman power in the organisation, controlling a total of 34 relationships. This is also reflected in that \citet{Krackhardt1987} highlighted manager 15 as an important agent in the organisational advice network. However, Node 15 does not have the highest betweenness or Bonacich centralities. Instead, Node 18 is the most prominent in terms of Bonacich and betweenness scores, although not being a middleman in the network.\footnote{The high betweenness and Bonacich centality measures might be a function of the in- and out-degree of Node 18.}

The reported Bonacich and betweenness centrality measures for the Krackhardt network confirm that both seem to be poor indicators for ranking middlemen: Node 15 is the most powerful middleman in the network, but has a Bonacich and betweenness centrality lower than Node 21. The Bonacich influence model does not consider the fact that middlemen are potentially able to exploit their position by using information from others and blocking the transmission of certain information and ideas.

\subsection{The elite Florentine marriage network}

The marriage network of elite houses in renaissance Florence has been used extensively to assess the effectiveness of many centrality measures to highlight positions of importance and influence \citep{Newman2003betweenness}. It has been shown in contemporary studies of this network that the Medici house had the highest centrality across a number of measures \citep[Chapter 2]{Jackson2008}. However, \citet{Roover1963}, \citet{Padgett1994}, and \citet{Goldthwaite2009} explain that their prominence in the marriage network is derived from their ability to access diverse sources of information within the Florentine aristocracy. As such, the Medici family filter information by choosing to allow or disallow information to spread between factions; thus largely monopolising inter-factional informational spread. Powerful brokerage opportunities, which the Medici took advantage of, emerged due to the inherent ``network disjunctures within the elite'' \citep[p.~1259]{Padgett1993}. In particular, \citet{Padgett1993} show that Cosimo de'Medici was able to gain access to, and control of, the flow of diverse information between opposing political factions and also between houses in the same faction. Within our context, its middleman position allowed the Medici family to attain power within Florentine society; especially, the Medici's ability to act as broker between a large number of houses crossing opposing political factions.

Keeping with the format and structure initially provided by \citet[p.~1276--77]{Padgett1993}---but unlike more recent renditions of the Florentine marriage network that are presented in terms of an undirected graph---we represent this network as a directed graph. An arc drawn from house $i$ to house $j$ illustrates a female from house or family $i$ married to a male in house $j$. The resulting marriage network is depicted in Figure~\ref{Fig:FlorentineFamilies}.\footnote{The data was initially gathered by \citet{Kent1978} and a block model network was constructed and used in \citet{Padgett1993} and \citet{Padgett1994}. The network provided in Figure~\ref{Fig:FlorentineFamilies} is directly derived from these studies. Both provide a rich analysis of the houses in Florence at this time.} Information flowed through these relationships, and marriages have often supported economic relationships in the form of trade, employment and loan provision~\citep{Kent2009}.

\begin{figure}[h]
\centering
\includegraphics[scale=0.37]{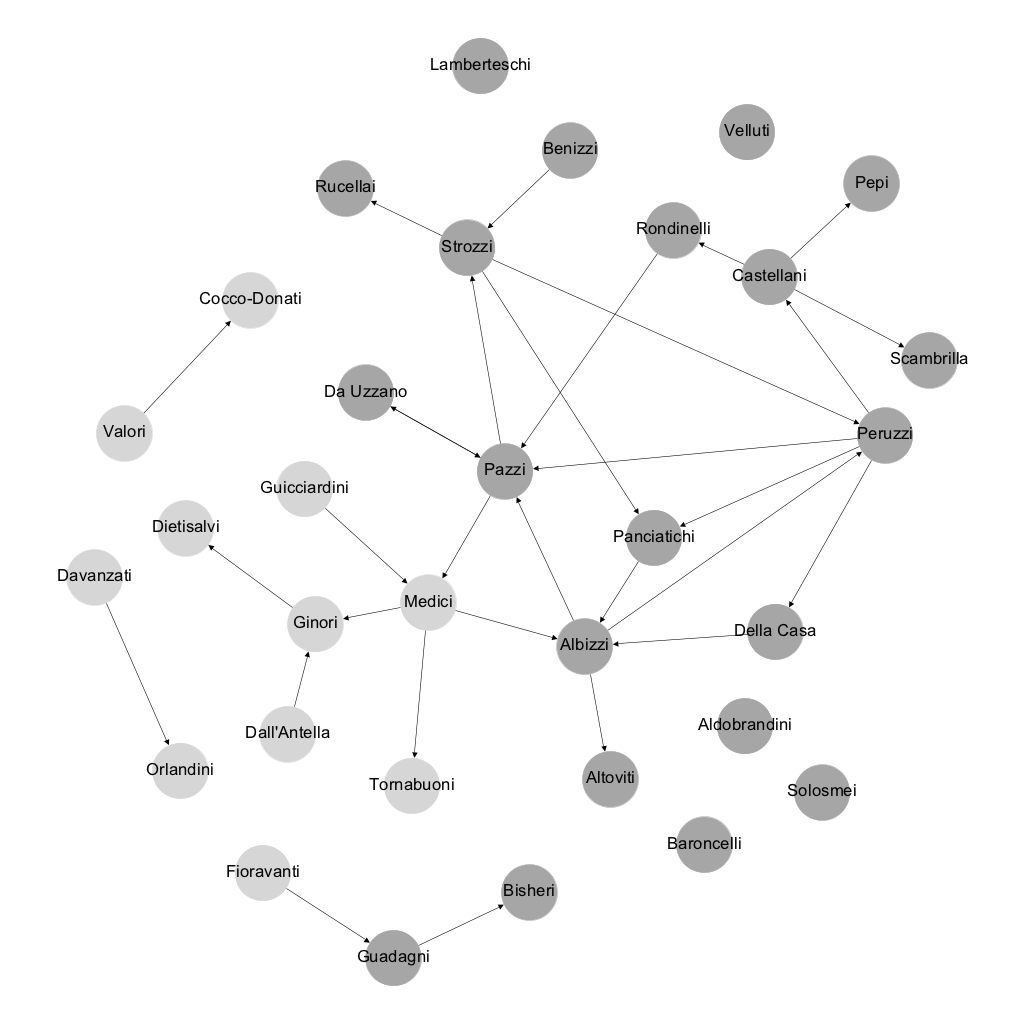}
\caption{Directed network of Florentine marriages (c. 1434)}
\label{Fig:FlorentineFamilies}
\end{figure}

Unlike more recent assessments of the Florentine marriage network we include houses as group of families; thus all nodes in the network represent a group of families under the same name. These houses are coloured depending on the factions that the houses were affiliated: light grey nodes are houses affiliated with the Medician faction and dark grey nodes are houses affiliated with the opposing Oligarch faction.

\paragraph{Network structure.}

There are 32 houses in this network: 11 houses in the Medician faction and 21 houses in the oligarchic faction. As depicted, the marriage network consists of 9 weakly connected components. The giant weakly connected component contains 62.5\% of all houses in the analysis. The network diameter is 6 with an average path length of $2.881$ and the directed network density is $0.031$. The degree distribution is similar to that of a power law. The maximal $k$-core is 4 and consists of the Pazzi, Peruzzi, Strozzi, and Albizzi houses.

The network produced in Figure~\ref{Fig:FlorentineFamilies} highlights a clear distinction between the connectivity of the two main factions. Only 3 out of 31 reported marriages were between houses in different factions. As such, it is clear that the Pazzi, Albizzi and Medici houses act as gate-keepers of the main information flows between the factions within the giant component. Each of these three houses acts as strong middlemen; however, the Medici is the only house that acts as a strong middleman between both factions. Indeed, the removal of the Medici leads to the partitioning of the factions in the giant component: The removal of either the Pazzi and Albizzi families does not have the same effect.

\paragraph{Centrality and power.}

We analyse the marriage network using in-degree $\left(d^-\right)$, out-degree $\left(d^+\right)$, \citet{Bonacich1987} eigenvector centrality ($E$), betweenness centrality ($BC$), and the normalised middleman power ($\nu$) of each house. The results of the analysis on middlemen is presented in Table~\ref{tabFlorence}. A full analysis of the centrality measures of all 32 houses in the network can be found in the table in the appendix to this paper.

\begin{table}
\begin{center}
\begin{tabular}{|l|ccccc|}
\toprule
House 	& $d^+_i$ & $d^-_i$ & $BC_i$    & $E_i$   & $\nu_i$  \\
\midrule
Albizzi** & 3 & 3 						 & 0.066 	 & 0.701   & 0.357 	  \\
Castellani**  & 3 & 1 						 & 0.034 	 & 0.310   & 0.187 	  \\
Ginori** & 1 & 2 						 & 0.014 	 & 0.257   & 0.076    \\
Guadagni** & 1 & 1 						 & 0.001 	 & 0.001   & 0.006 	  \\
Medici** & 3 & 2 						 & 0.058 	 & 0.506   & 0.269 	  \\
Pazzi**  & 3 & 4 						 & 0.093 	 & 1.000   & 0.503    \\
Peruzzi*  & 4 & 2 						 & 0.070 	 & 0.612   & 0.287    \\
Rondinelli*  & 1 & 1 						 & 0.014 	 & 0.157   & 0.076    \\
Strozzi** & 3 & 2 						 & 0.053 	 & 0.506   & 0.152    \\ 
\bottomrule
\end{tabular}
\caption{Middlemen in the directed Florentine marriage network}
\label{tabFlorence}
\end{center}
\end{table}

The simplest measurement of node centrality is the the number of direct successors---represented by the out-degree ($d^+_i$)---and the number of direct predecessors---represented by the in-degree ($d^-_i$)---of a node. The Peruzzi has the greatest number of direct successors (4) and the Pazzi have the greatest number of direct predecessors (4); the Pazzi also have the greatest sum of direct successors and direct predecessors (7).\footnote{In an undirected network the connections that a node has is given by the total number of connections that a node has; this is not necessarily the same as the sum of direct predecessors and direct successors as some node may be both a direct predecessor and direct successor.}

If the network were represented as an undirected network, the Pazzi, Albizzi, and Peruzzi are connected to 6 houses each. The Medici and Strozzi are connected to 5 houses each. There are two aspects to note from the assessment of node degree. First we note that, in general, families with a relatively higher degree are more prone to be middlemen in the network. This is true for most families, apart from the Guadagni who are conveniently positioned such that their single in-degree and single out-degree forms a middleman position. Second, we note that the Medici faction does not have the highest degree centrality relative to other families in the Oligarchic faction.

The Pazzi possesses the highest betweenness centrality in the directed network. The measure can favour middlemen thus typically ranking them higher; eight of the top ten houses in terms of their betweenness score $BC$ are either weak or strong middlemen. The Medici have the highest betweenness centrality ($0.166$) followed by the Pazzi ($0.142$). Unsurprisingly, the Bonacich centrality measure ranks the Pazzi highly, but also ranks many non-middlemen highly; specifically the Panciatichi house. From the network topology alone there is no indication to suggest that the Panciatichi should have had a prominent role in the Florentine aristocracy. Although the Medici rank highly with this measure, the relevance of an eigenvector centrality measure is questionable: there is no real reason to believe why the importance of a house would come from its degree and the degree of its neighbours alone.

Assessing the marriage network with the middleman power measure highlights the Medici family as a strong middleman, along with the prominent Albizzi, Castellani, Ginori, Pazzi, and Strozzi houses. It is, however, the Pazzi and Albizzi that have a greater middleman power measure ($0.503$ and $0.357$ respectively) than the Medici ($0.269$). The diversity of the Medici's brokered relationships extend further than those of the Albizzi and Pazzi as the Medici brokers between factions. If, however, the directed marriage network were converted into an undirected network such that information can flow in both directions, then the Medici becomes the most powerful middleman with a normalised middleman power of $0.470$; they are followed by the Ginori, Castellani and Strozzi families who have a normalised middleman power of $0.213$.


\paragraph{Middleman robustness.}


The robustness of each middleman position is measured in terms of the $\rho_{i}$-robustness, (dual) $\rho^{\star}_{i}$-robustness, and node $\psi_{i}$-robustness measures, reported in Table~\ref{FlorenceRobust}. In general we find that the $\rho^{\star}_{i}$- and $\psi_{i}$-robustness measures provide identical results. The directed network highlights the Albizzi, Pazzi, and Medici as being robust; especially in terms of the $\rho_{i}$-robustness measure. The results also suggests that there is a clear distinction between the consistency of some middlemen over others. The Medici is highlighted as being the most robust strong middleman in the undirected network. Notably, even though the Pazzi and Albizzi were robust in the directed network they both perform poorly in terms of the undirected representation of the marriage network.

\begin{table}
\centering
\begin{tabular}{|l|ccc|ccc|}
\toprule
  & \multicolumn{3}{|c|}{Directed}   & \multicolumn{3}{|c|}{Undirected}      \\
House   & $\rho_{i}$ & $\rho^{\star}_{i}$ & $\psi_{i}$ & $\rho_{i}$ & $\rho^{\star}_{i}$ & $\psi_{i}$ \\
\midrule
Albizzi**    & 4            & 3           & 3            & 1            & 1           & 1            \\
Castellani** & 3            & 3           & 3            & 2            & 2           & 2            \\
Ginori**     & 2            & 1           & 1            & 2            & 2           & 2            \\
Guadagni**   & 1            & 1           & 1            & 1            & 1           & 1            \\
Medici**     & 4            & 2           & 2            & 3            & 3           & 3            \\
Pazzi**      & 6            & 3           & 3            & 1            & 1           & 1            \\
Peruzzi*     & 2            & 2           & 2            & 0            & 0           & 0            \\
Rondinelli*  & 1            & 1           & 1            & 0            & 0           & 0            \\
Strozzi**    & 4            & 2           & 2            & 2            & 2           & 2            \\
\bottomrule
\end{tabular}
\caption{Robustness of middlemen in the Florentine marriage network}
\label{FlorenceRobust}
\end{table}

\medskip \noindent Centrality and robustness measures are purely topological and do not highlight the heterogeneous factions that exist within this network and the importance of information brokerage between them. Therefore they assume that relationships can be created and severed without social or institutional pressures. As a consequence it could be argued that, regardless of these robustness measures, the Medici has the highest middleman robustness due to the fact that they are weak and strong middlemen across both factions. Indeed, social and marriage relationships cannot be formed so seamlessly between families of opposing factions, therefore the robustness of the Medici's position is further strengthened due to the societal environment.


\section{Concluding remarks}
\label{sec:Conclusion}

Middlemen possess the ability to connect pairs of nodes who would otherwise be disconnected; this can have liberating externalities for those directly or indirectly connected to middlemen in the form of opening new exchange routes and channels of information. On the other hand, middlemen can exploit their position and, as such, extract positional rents. Whether a middleman behaves in an exploitative or facilitative is ambiguous and depends on the institutional governance of activities in the network.

Our middleman power measure provides a tool that measures the positional power of the middleman and, therefore, is an objective quantifier of the extractive and value-generating abilities of the middleman. We also note that the middleman power measure should not be considered as a replacement for other centrality measures. It is itself not just a measure of centrality; rather it identifies a certain type of node in a network and measures brokerage. The measure should be complimented with other measures of centrality.


\singlespace

\bibliographystyle{agsm}
\bibliography{OSDB}

\newpage

\appendix

\begin{center}

\section*{Appendix: Centrality in the Florentine network c.~1434}
\label{AppendixA}

\end{center}

\begin{table}[h]
\begin{center}
\begin{tabular}{|l|c|ccccc|}
\toprule
House & Faction & $d^+_i$ & $d^-_i$ & $BC_i$	&  $E_i$	&  $\nu_i$	\\
\midrule
Albizzi**       & Oligarch	& 3 & 3 & 0.066 & 0.701 & 0.357      \\
Aldobrandini    & Oligarch	& 0 & 0 & 0.000 & 0.000 & 0.000      \\
Altoviti        & Oligarch	& 0 & 1 & 0.000 & 0.355 & 0.000      \\
Baroncelli      & Oligarch	& 0 & 0 & 0.000 & 0.000 & 0.000      \\
Benizzi         & Oligarch	& 1 & 0 & 0.000 & 0.000 & 0.000      \\
Bisheri         & Oligarch	& 0 & 1 & 0.000 & 0.002 & 0.000      \\
Castellani**    & Oligarch	& 3 & 1 & 0.034 & 0.310 & 0.187      \\
C-Donati        & Medician	& 0 & 1 & 0.000 & 0.001 & 0.000      \\
Da Uzzano       & Oligarch	& 1 & 1 & 0.000 & 0.506 & 0.000      \\
Dall'Antella    & Medician	& 1 & 0 & 0.000 & 0.000 & 0.000      \\
Davanzati       & Medician	& 1 & 0 & 0.000 & 0.000 & 0.000      \\
Della Casa      & Oligarch	& 1 & 1 & 0.001 & 0.309 & 0.000      \\
Dietisalvi      & Medician	& 0 & 1 & 0.000 & 0.132 & 0.000      \\
Fioravanti      & Medician	& 1 & 0 & 0.000 & 0.000 & 0.000      \\
Ginori(**)      & Medician	& 1 & 2 & 0.014 & 0.257 & 0.076      \\
Guadagni**      & Oligarch	& 1 & 1 & 0.001 & 0.001 & 0.006      \\
Guicciardini    & Medician	& 1 & 0 & 0.000 & 0.000 & 0.000      \\
Lamberteschi    & Oligarch	& 0 & 0 & 0.000 & 0.000 & 0.000      \\
Medici**        & Medician	& 3 & 2 & 0.058 & 0.506 & 0.269      \\
Orlandini       & Medician	& 0 & 1 & 0.000 & 0.001 & 0.000      \\
Panciatichi     & Oligarch	& 1 & 2 & 0.005 & 0.566 & 0.000      \\
Pazzi**         & Oligarch    & 3 & 4 & 0.093 & 1.000 & 0.503      \\
Pepi            & Oligarch	& 0 & 1 & 0.000 & 0.157 & 0.000      \\
Peruzzi*        & Oligarch	& 4 & 2 & 0.070 & 0.612 & 0.287      \\
Rondinelli*     & Oligarch	& 1 & 1 & 0.014 & 0.157 & 0.076      \\
Rucellai        & Oligarch	& 0 & 1 & 0.000 & 0.256 & 0.000      \\
Scambrilla      & Oligarch	& 0 & 1 & 0.000 & 0.157 & 0.000      \\
Solosmei        & Oligarch	& 0 & 0 & 0.000 & 0.000 & 0.000      \\
Strozzi**       & Oligarch	& 3 & 2 & 0.053 & 0.506 & 0.152      \\
Tornabuoni      & Medician	& 0 & 1 & 0.000 & 0.257 & 0.000      \\
Valori          & Medician	& 1 & 0 & 0.000 & 0.000 & 0.000      \\
Velluti         & Oligarch	& 0 & 0 & 0.000 & 0.000 & 0.000      \\
\bottomrule
\end{tabular}
\label{tabFlorenceA}
\end{center}
\end{table}

\end{document}